\documentclass{article}

\usepackage{graphicx}
\usepackage{color}
\usepackage{stmaryrd}
\usepackage[all]{xy}

\usepackage[numbers]{natbib}
\usepackage{natbibspacing}
\renewcommand{\bibfont}{\small}

\definecolor{myurlcolor}{rgb}{0,0,0.5}
\usepackage{hyperref}
\hypersetup{colorlinks,
linkcolor=black,
citecolor=black,
urlcolor=myurlcolor}

\usepackage{tocloft}
\usepackage{latexsym}
\usepackage{amssymb,amsmath}
\newcommand{\bref}[1]{(\ref{#1})}
\newcommand{\such}{:}
\newcommand{\epsln}{\varepsilon}
\newcommand{\R}{\mathbb{R}}
\newcommand{\demph}[1]{\textbf{#1}}
\newcommand{\sub}{\subseteq}
\newcommand{\cell}[4]{\put(#1,#2){\makebox(0,0)[#3]{#4}}}
\newcommand{\mg}[1]{\left| #1 \right|}
\newcommand{\transp}{T}
\DeclareMathOperator{\supp}{supp}
\newcommand{\vecstyle}[1]{\mathbf{#1}}
\newcommand{\mxstyle}[1]{\mathbf{#1}}
\newcommand{\p}{\vecstyle{p}}
\newcommand{\rvec}{\vecstyle{r}}
\newcommand{\e}{\vecstyle{e}}
\newcommand{\s}{\vecstyle{s}}
\newcommand{\w}{\vecstyle{w}}
\newcommand{\x}{\vecstyle{x}}
\newcommand{\Z}{\mxstyle{Z}}
\newcommand{\I}{\mxstyle{I}}
\newcommand{\M}{\mxstyle{M}}

\newcommand{\pext}[1]{\p(#1)} 
\newcommand{\primvv}[3]{{}^{#1}\!{#3}^{#2}}
\newcommand{\dvv}[2]{\primvv{#1}{#2}{D}}
\newcommand{\dqz}{\dvv{q}{\Z}}

\newcommand{\pmax}{\p_{\text{max}}}
\newcommand{\Dmax}[1]{D_\textup{max}(#1)}
\newcommand{\pcolor}{\color[rgb]{0,0,0.5}} 
\newcommand{\pmaxcolor}{\color[rgb]{0,0.5,0}} 
\newcommand{\pprimecolor}{\color[rgb]{0.75,0,0}} 
\newcommand{\adjt}{\mathrel{\sim}}
\newcommand{\adjto}[2]{#1 \adjt #2}
\newcommand{\adjtoin}[3]{#1 \adjt #2 \text{ in } #3}
\newcommand{\gcomp}[1]{\overline{#1}}
\newcommand{\cliqueno}[1]{\omega(#1)}
\newcommand{\csuch}{\colon}
\newcommand{\gedge}{\!\!-\!\!}
\newcommand{\rre}[3]{I_{#1}(#2\mathbin{|}#3)}

\usepackage[amsmath,thmmarks]{ntheorem}

\newtheorem{thm}{Theorem}[section]
\newtheorem{propn}[thm]{Proposition}
\newtheorem{lemma}[thm]{Lemma}
\newtheorem{cor}[thm]{Corollary}

\theorembodyfont{\normalfont}

\newtheorem{defn}[thm]{Definition}
\newtheorem{example}[thm]{Example}

\newtheorem{remark}[thm]{Remark}

\theoremstyle{nonumberplain}
\theoremsymbol{\ensuremath{\square}}
\qedsymbol{\ensuremath{\square}}

\newtheorem{proof}{Proof}

\title{Maximizing diversity in biology and beyond}
\author{Tom Leinster%
\thanks{School of Mathematics, University of Edinburgh, UK, and Boyd Orr
  Centre for Population and Ecosystem Health, University of Glasgow, UK;
  Tom.Leinster@ed.ac.uk}
\qquad
Mark W. Meckes%
\thanks{Department of Mathematics, Applied Mathematics, and Statistics,
  Case Western Reserve University, Cleveland, Ohio,
USA;
  Mark.Meckes@case.edu}}
\date{}

\begin{document}

\sloppy
\maketitle

\begin{abstract}
Entropy, under a variety of names, has long been used as a measure of
diversity in ecology, as well as in genetics, economics and other fields.
There is a spectrum of viewpoints on diversity, indexed by a real parameter
$q$ giving greater or lesser importance to rare species.  Leinster and
Cobbold~\cite{MDISS} proposed a one-parameter family of diversity measures
taking into account both this variation and the varying similarities
between species.  Because of this latter feature, diversity is not
maximized by the uniform distribution on species.  So it is natural to ask:
which distributions maximize diversity, and what is its maximum value?%
\\%
\hspace*{\parindent}%
In principle, both answers depend on $q$, but our main theorem is that
neither does.  Thus, there is a single distribution that maximizes
diversity from all viewpoints simultaneously, and any list of species has
an unambiguous maximum diversity value.  Furthermore, the maximizing
distribution(s) can be computed in finite time, and any distribution
maximizing diversity from some particular viewpoint $q > 0$ actually
maximizes diversity for all $q$.%
\\%
\hspace*{\parindent}%
Although we phrase our results in ecological terms, they apply very
widely, with applications in graph theory and metric geometry.
\end{abstract}

\vspace*{-1ex}%
\hspace*{.035\textwidth}\parbox{.85\textwidth}{\tableofcontents}%
\vspace*{1ex}

\section{Introduction}
\label{sec:intro}

For decades, ecologists have used entropy-like quantities as measures of
biological diversity.  The basic premise is that given a biological
community or ecosystem containing $n$ species in proportions $p_1,
\ldots, p_n$, the entropy of the probability distribution $(p_i)$ indicates
the extent to which the community is balanced or `diverse'.  Shannon
entropy itself is often used; so too are many variants, as we shall see.
But almost all of them share the property that for a fixed number $n$ of
species, the entropy is maximized by the uniform distribution $p_i = 1/n$.

However, there is a growing appreciation that this crude model of a
biological community is too far from reality, in that it takes no notice of
the varying similarities between species.  For instance, we would
intuitively judge a meadow to be more diverse if it consisted of ten
dramatically different plant species than if it consisted of ten species of
grass.  This has led to the introduction of measures that do take into
account inter-species similarities~\cite{RaoDDC,MDISS}.  In mathematical
terms, making this refinement essentially means extending the classical
notion of entropy from probability distributions on a finite set to
probability distributions on a finite metric space.

The maximum entropy problem now becomes more interesting.  Consider, for
instance, a pond community consisting of two very similar species of frog
and one species of newt.  We would not expect the maximum entropy (or
diversity) to be achieved by the uniform distribution $(1/3, 1/3, 1/3)$,
since the community would then be $2/3$ frog and only $1/3$ newt.  We might
expect the maximizing distribution to be closer to $(1/4, 1/4, 1/2)$;
the exact answer should depend on the degrees of similarity of the species
involved.  We return to this scenario in Example~\ref{eg:three}.

For the sake of concreteness, this paper is written in terms of an
ecological scenario: a community of organisms classified into species.
However, nothing that we do is intrinsically ecological, or indeed
connected to any specific branch of science.  Our results apply equally to any
collection of objects classified into types.

It is well understood that Shannon entropy is just one point (albeit a
special one) on a continuous spectrum of entropies, indexed by a parameter
$q \in [0, \infty]$.  This spectrum has been presented in at least two
ways: as the R\'enyi entropies $H_q$~\cite{Reny} and as the so-called
Tsallis entropies $S_q$ (actually introduced as biodiversity measures by
Patil and Taillie prior to Tsallis's work in physics, and earlier still in
information theory \cite{Tsal,PaTaDCM,HaCh}):
\[
H_q(\p)
=
\frac{1}{1 - q} \log \sum_{i = 1}^n p_i^q,
\qquad
S_q(\p)
=
\frac{1}{q - 1} \biggl( 1 - \sum_{i = 1}^n p_i^q \biggr).
\]
Both $H_q$ and $S_q$ converge to Shannon entropy as $q \to 1$.  Moreover,
$H_q$ and $S_q$ can be obtained from one another by an increasing
invertible transformation, and in this sense are interchangeable.

When $H_q$ or $S_q$ is used as a diversity measure, $q$ controls the weight
attached to rare species, with $q = 0$ giving as much importance to rare
species as common ones, and the limiting case $q = \infty$ reflecting only
the prevalence of the most common species.  Different values of $q$ produce
genuinely different judgements on which of two distributions is the more
diverse.  For instance, if over time a community loses some species but
becomes more balanced, then the R\'enyi and Tsallis entropies decrease for
$q = 0$ but increase for $q = \infty$.  Varying $q$ therefore allows us to
incorporate a spectrum of viewpoints on the meaning of the word
`diversity'.

Here we use the diversity measures introduced by Leinster and
Cobbold~\cite{MDISS}, which both (i)~reflect this spectrum of viewpoints by
including the variable parameter $q$, and (ii)~take into account the
varying similarities between species.  We review these measures in Sections
\ref{sec:spectrum}--\ref{sec:measures}.  In the extreme case where
different species are assumed to have nothing whatsoever in common, they
reduce to the exponentials of the R\'enyi entropies, and in other special
cases they reduce to other diversity measures used by ecologists.  In
practical terms, the measures of~\cite{MDISS} have been used to assess a
variety of ecological systems, from communities of
microbes~\cite{VPDL,BCMV} and crustacean zooplankton~\cite{JTYP} to alpine
plants~\cite{CMLT} and arctic predators~\cite{BRBT}, as well as being
applied in non-biological contexts such as computer network
security~\cite{WZJS}.

Mathematically, the set-up is as follows.  A biological community is
modelled as a probability distribution $\p = (p_1, \ldots, p_n)$
(representing the proportions of species) together with an $n \times n$
matrix $\Z$ (whose $(i, j)$-entry represents the similarity between species
$i$ and $j$).  From this data, a formula gives a real number $\dqz(\p)$ for
each $q \in [0, \infty]$, called the `diversity of order $q$' of the
community.  As for the R\'enyi entropies, different values of $q$ make
different judgements: for instance, it may be that for two distributions
$\p$ and $\p'$,
\[
\dvv{1}{\Z}(\p) < \dvv{1}{\Z}(\p')
\qquad 
\text{but} 
\qquad
\dvv{2}{\Z}(\p) > \dvv{2}{\Z}(\p').
\]

Now consider the maximum diversity problem.  Fix a list of species whose
similarities to one another are known; that is, fix a matrix $\Z$ (subject
to hypotheses to be discussed).  The two basic questions are:
\begin{itemize}
\item 
Which distribution(s) $\p$ maximize the diversity $\dqz(\p)$ of order $q$?

\item
What is the value of the maximum diversity $\sup_\p \dqz(\p)$?
\end{itemize}
This can be interpreted ecologically as follows: if we have a fixed list of
species and complete control over their abundances within our community,
how should we choose those abundances in order to maximize the diversity,
and how large can we make that diversity?  

In principle, both answers depend on $q$.  After all, we have seen that if
distributions are ranked by diversity then the ranking varies according to
the value of $q$ chosen.  But our main theorem is that, in fact, both
answers are \emph{independent} of $q$:

\begin{trivlist}\item%
\textbf{Theorem~\ref{thm:main} (Main theorem)}\hspace*{1em}\itshape There
exists a probability distribution on $\{1, \ldots, n\}$ that maximizes
$\dqz$ for all $q \in [0, \infty]$.  Moreover, the maximum diversity
$\sup_\p \dqz(\p)$ is independent of $q \in [0, \infty]$.
\end{trivlist}

So, there is a `best of all possible worlds': a distribution
that maximizes diversity no matter what viewpoint one takes on the relative
importance of rare and common species.

This theorem merely asserts the \emph{existence} of a maximizing
distribution.  However, a second theorem describes how to \emph{compute}
all maximizing distributions, and the maximum diversity, in a finite number
of steps (Theorem~\ref{thm:comp}).  

Better still, if by some means we have found a distribution $\p$ that
maximizes the diversity of \emph{some} order $q > 0$, then a further result
asserts that $\p$ maximizes diversity of \emph{all} orders
(Corollary~\ref{cor:irrelevance}).  For instance, it is often easiest to
find a maximizing distribution for diversity of order $\infty$ (as in
Example~\ref{eg:naive} and Proposition~\ref{propn:graph-main}), and it is
then automatic that this distribution maximizes diversity of all orders.

Let us put these results into context.  First, they belong to the huge body
of work on maximum entropy problems.  For example, the normal distribution
has the maximum entropy among all probability distributions on $\R$ with a
given mean and variance, a property which is intimately connected with its
appearance in the central limit theorem.  This alone would be enough
motivation to seek maximum entropy distributions in other settings (such as
the one at hand), quite apart from the importance of maximum entropy in
thermodynamics, machine learning, macroecology, and so on.

Second, we will see that maximum diversity is very closely related to the
emerging invariant known as magnitude.  This is defined in the extremely
wide generality of enriched category theory \cite[Section~1]{MMS}, and
specializes in interesting ways in a variety of mathematical fields.  For
instance, it automatically produces a notion of the Euler characteristic of
an (ordinary) category, closely related to topological Euler
characteristic~\cite{ECC}; in the context of metric spaces, magnitude
encodes geometric information such as volume and dimension
\cite{BaCa,MeckMDC,WillMSS}; in graph theory, magnitude is a new invariant
that turns out to be related to a graded homology theory for
graphs~\cite{MG,HeWi}; and in algebra, magnitude produces an invariant of
associative algebras that can be understood as a homological Euler
characteristic~\cite{MFDA}.

This work is self-contained.  To that end, we begin by explaining and
defining the diversity measures in~\cite{MDISS} (Sections
\ref{sec:spectrum}--\ref{sec:measures}).  Then come the results: 
preparatory lemmas in Section~\ref{sec:invt}, and the main results in
Sections~\ref{sec:main} and~\ref{sec:comp}.  Examples are given in Sections
\ref{sec:examples-simple}--\ref{sec:examples-pos-def}, including results on
special cases such as when the similarity matrix $\Z$ is either the
adjacency matrix of a graph or positive definite.  Perhaps
counterintuitively, a distribution that maximizes diversity can eliminate
some species entirely.  This is addressed in Section~\ref{sec:elim}, where
we derive necessary and sufficient conditions on $\Z$ for maximization to
preserve all species.  Finally, we state some open questions
(Section~\ref{sec:questions}).

The main results of this paper previously appeared in a preprint of
Leinster~\cite{METAMB}, but the proofs we present here are substantially
simpler.  Of the new results, Lemma~\ref{lemma:elim-main} (the key to
our results on preservation of species by maximizing distributions) borrows
heavily from an argument of Fremlin and Talagrand~\cite{FrTa}.

\paragraph{Conventions} 
\label{p:conventions}
A vector $\x = (x_1, \ldots, x_n) \in \R^n$ is \demph{nonnegative} if $x_i
\geq 0$ for all $i$, and \demph{positive} if $x_i > 0$ for all
$i$.  The \demph{support} of $\x \in \R^n$ is
\[
\supp(\x) 
= 
\bigl\{ i \in \{1, \ldots, n\} \such x_i \neq 0 \bigr\},
\]
and $\x$ has \demph{full support} if $\supp(\x) = \{1, \ldots, n\}$.  
A real symmetric $n \times n$ matrix $\Z$ is \demph{positive semidefinite}
if $\x^\transp \Z \x \geq 0$ for all $\vecstyle{0} \neq \x \in \R^n$, and
\demph{positive definite} if this inequality is strict.

\section{A spectrum of viewpoints on diversity}
\label{sec:spectrum}

Ecologists began to propose quantitative definitions of biological
diversity in the mid-twentieth century \cite{SimpMD,WhitVSM}, setting in
motion more than fifty years of heated debate, dozens of further proposed
diversity measures, hundreds of scholarly papers, at least one book devoted
to the subject~\cite{Magu}, and consequently, for some, despair (already
expressed by 1971 in a famously-titled paper of Hurlbert~\cite{Hurl}).
Meanwhile, parallel discussions were taking place in disciplines such as
genetics~\cite{KiCr}, economists were using the same formulas to measure
wealth inequality and industrial concentration~\cite{HaKa}, and information
theorists were developing the mathematical theory of such quantities under
the name of entropy rather than diversity.

Obtaining accurate data about an ecosystem is beset with practical and
statistical problems, but that is not the reason for the prolonged debate.
Even assuming that complete information is available, there are genuine
differences of opinion about what the word `diversity' should mean.  We
focus here on one particular axis of disagreement, illustrated by the
examples in Figure~\ref{fig:birds}.

\begin{figure}
\begin{center}
\setlength{\unitlength}{1mm}
\begin{picture}(120,38)
\cell{30}{0}{b}{\includegraphics[width=40\unitlength]{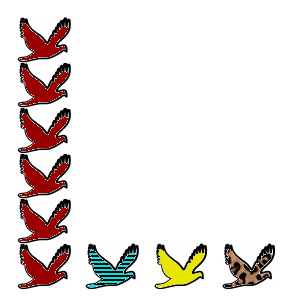}}
\cell{30}{-1}{t}{(a)}
\cell{95}{0}{b}{\includegraphics[width=30\unitlength]{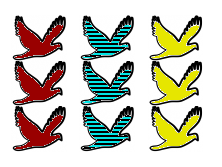}}
\cell{95}{-1}{t}{(b)}
\end{picture}
\end{center}
\caption{Two bird communities.  Heights of stacks indicate species
  abundances.  In~(a), there are four species, with the first dominant and
  the others relatively rare.  In~(b), the fourth species is absent but the
  community is otherwise evenly balanced.}
\label{fig:birds}
\end{figure}

One extreme viewpoint on diversity is that preservation of species is all
that matters: `biodiversity' simply means the number of species present (as
is common in the ecological literature as well as the media).  Since no
attention is paid to the abundances of the species present, rare species
count for exactly as much as common species.  From this viewpoint,
community~(a) of Figure~\ref{fig:birds} is more diverse than community~(b),
simply because it contains more species.

The opposite extreme is to ignore rare species altogether and consider only
those that are most common.  (This might be motivated by a focus on
overall ecosystem function.)  From this viewpoint, community~(b) is more
diverse than community~(a), because it is better-balanced: (a)~is dominated
by a single common species, whereas (b)~has three common species in equal
proportions.

Between these two extremes, there is a spectrum of intermediate viewpoints,
attaching more or less weight to rare species.  Different scientists have
found it appropriate to adopt different positions on this spectrum for
different purposes, as the literature amply attests.

Rather than attempting to impose one particular viewpoint, we will consider
all equally.  Thus, we use a one-parameter family of diversity measures,
with the `viewpoint parameter' $q \in [0, \infty]$ controlling one's
position on the spectrum.  Taking $q = 0$ will give rare species as much
importance as common species, while taking $q = \infty$ will give rare
species no importance at all.

There is one important dimension missing from the discussion so far.  We
will consider not only the varying abundances of the species, but also the
varying similarities between them.  This is addressed in the next section.

\section{Distributions on a set with similarities}
\label{sec:model}

In this section and the next, we give a brief introduction to the diversity
measures of Leinster and Cobbold~\cite{MDISS}.  We have two tasks.  We must
build a mathematical model of the notion of `biological community' (this
section).  Then, we must define and explain the diversity measures
themselves (next section).

In brief, a biological community will be modelled as a finite set (whose
elements are the species) equipped with both a probability distribution
(indicating the relative abundances of the species) and, for each pair of
elements of the set, a similarity coefficient (reflecting the similarities
between species).  

Let us now consider each of these aspects in turn.  First, we assume a
community or system of individuals, partitioned into $n \geq 1$ species.
The word `species' need not have its standard meaning: it can denote any
unit thought meaningful, such as genus, serotype (in the case of viruses),
or the class of organisms having a particular type of diet.  It need not
even be a biological grouping; for instance, in~\cite{McMi} the units are
soil types.  For concreteness, however, we write in terms of an ecological
community divided into species.  The division of a system into species or
types may be somewhat artificial, but this is mitigated by the introduction
of the similarity coefficients (as shown in~\cite{MDISS}, p.482).

Second, each species has a relative abundance, the proportion of organisms
in the community belonging to that species.  Thus, listing the species in
order as $1, \ldots, n$, the relative abundances determine a vector $\p =
(p_1, \ldots, p_n)$.  This is a probability distribution: $p_i \geq 0$ for
each species $i$, and $\sum_{i = 1}^n p_i = 1$.  Abundance can be measured
in any way thought relevant, e.g.\ number of individuals, biomass, or (in
the case of plants) ground coverage.

Critically, the word `diversity' refers only to the \emph{relative}, not
absolute, abundances.  If half of a forest burns down, or if a patient
loses $90\%$ of their gut bacteria, then it may be an ecological or medical
disaster; but assuming that the system is well-mixed, the diversity does
not change.  In the language of physics, diversity is an intensive quantity
(like density or temperature) rather than an extensive quantity (like mass
or heat), meaning that it is independent of the system's size.

The third and final aspect of the model is inter-species similarity.  For
each pair $(i, j)$ of species, we specify a real number $Z_{ij}$
representing the similarity between species $i$ and $j$.  This defines an
$n \times n$ matrix $\Z = (Z_{ij})_{1 \leq i, j \leq n}$.  In~\cite{MDISS},
similarity is taken to be measured on a scale of $0$ to $1$, with $0$
meaning total dissimilarity and $1$ that the species are identical.  Thus,
it is assumed there that%
\begin{equation}
\label{eq:stronger-sim-hyps}
0 \leq Z_{ij} \leq 1 \text{ for all } i, j,
\qquad
Z_{ii} = 1 \text{ for all } i.
\end{equation}
In fact, our maximization theorems will only require the weaker hypotheses
\begin{equation}
\label{eq:weaker-sim-hyps}
Z_{ij} \geq 0 \text{ for all } i, j,
\qquad
Z_{ii} > 0 \text{ for all } i
\end{equation}
together with the requirement that $\Z$ is a symmetric matrix.  (In the
appendix to~\cite{MDISS}, matrices satisfying~\eqref{eq:weaker-sim-hyps}
were called `relatedness matrices'.)

Just as the meanings of `species' and `abundance' are highly flexible, so
too is the meaning of `similarity':

\begin{example}
The simplest similarity matrix $\Z$ is the identity matrix $\I$.  This is
called the \demph{naive model} in~\cite{MDISS}, since it embodies the
assumption that distinct species have nothing in common.  Crude though this
assumption is, it is implicit in the diversity measures most popular
in the ecological literature \cite[Table~1]{MDISS}.  
\end{example}

\begin{example}
With the rapid fall in the cost of DNA sequencing, it is increasingly
common to measure similarity genetically (in any of several ways).  Thus,
the coefficients $Z_{ij}$ may be chosen to represent percentage genetic
similarities between species.  This is an effective strategy even when the
taxonomic classification is unclear or incomplete~\cite{MDISS}, as is often
the case for microbial communities~\cite{VPDL}.
\end{example}

\begin{example}
\label{eg:sim-phylo}
Given a suitable phylogenetic tree, we may define the similarity between
two present-day species as the proportion of evolutionary time before the
point at which the species diverged.
\end{example}

\begin{example}
\label{eg:sim-taxo}
In the absence of more refined data, we can measure species similarity
according to a taxonomic tree.  For instance, we might define
\[
Z_{ij}
=
\begin{cases}
1       &\text{if } i = j,      \\
0.8     &\text{if species } i \text{ and } j \text{ are different but of
  the same genus},  \\
0.5     &\text{if species } i \text{ and } j \text{ are of different genera
  but the same family},     \\
0       &\text{otherwise}.
\end{cases}
\]
\end{example}

\begin{example}
\label{eg:sim-metric}
In purely mathematical terms, an important case is where the similarity
matrix arises from a metric $d$ on the set $\{1, \ldots, n\}$ via the
formula $Z_{ij} = e^{-d(i, j)}$.  Thus, the community is modelled as a
probability distribution on a finite metric space.  (The naive model
corresponds to the metric defined by $d(i, j) = \infty$ for all $i \neq
j$.)  The diversity measures that we will shortly define can be understood
as (the exponentials of) R\'enyi-like entropies for such distributions. 
\end{example}

\section{The diversity measures}
\label{sec:measures}

Here we state the definition of the diversity measures of~\cite{MDISS},
which we will later seek to maximize.  We then explain the reasons for this
particular definition.

As in Section~\ref{sec:model}, we take a biological community modelled as a
finite probability distribution $\p = (p_1, \ldots, p_n)$ together with an
$n \times n$ matrix $\Z$ satisfying~\eqref{eq:weaker-sim-hyps}.  As
explained in Section~\ref{sec:spectrum}, we define not one diversity
measure but a family of them, indexed by a parameter $q \in [0, \infty]$
controlling the emphasis placed on rare species.  The \demph{diversity of
  order $q$} of the community is
\begin{equation}
\label{eq:the-measures}
\dqz(\p)
=
\biggl( 
\sum_{i \in \supp(\p)} p_i (\Z\p)_i^{q - 1} 
\biggr)^{1/(1 - q)}
\end{equation}
($q \neq 1, \infty$).  Here $\supp(\p)$ is the support of $\p$
(Conventions, p.\pageref{p:conventions}), $\Z \p$ is the column vector
obtained by multiplying the matrix $\Z$ by the column vector $\p$, and
$(\Z\p)_i$ is its $i$th entry.  The hypotheses~\eqref{eq:weaker-sim-hyps}
imply that $(\Z\p)_i > 0$ whenever $i \in \supp(\p)$, and so $\dqz(\p)$ is
well-defined.

Although this formula is invalid for $q = 1$, it converges as $q \to 1$,
and $\dvv{1}{\Z}(\p)$ is defined to be the limit.  The same is true for $q =
\infty$.  Explicitly,
\begin{align*}
\dvv{1}{\Z}(\p) &
=
\prod_{i \in \supp(\p)} (\Z\p)_i^{-p_i}
=
\exp \biggl( - \sum_{i \in \supp(\p)} p_i \log(\Z\p)_i \biggr),   \\
\dvv{\infty}{\Z}(\p)    &
=
1\Big/\max_{i \in \supp(\p)} (\Z\p)_i.
\end{align*}

The applicability, context and meaning of
definition~\eqref{eq:the-measures} are discussed at length in~\cite{MDISS}.
Here we briefly review the principal points.

First, the definition includes as special cases many existing quantities
going by the name of diversity or entropy.  For instance, in the naive
model $\Z = \I$, the diversity $\dvv{q}{\I}(\p)$ is the exponential of the
R\'enyi entropy of order $q$, and is also known in ecology as the Hill
number of order $q$.  (References for this and the next two paragraphs are
given in \cite[Table~1]{MDISS}.)

Continuing in the naive model $\Z = \I$ and specializing further to
particular values of $q$, we obtain other known quantities:
$\dvv{0}{\I}(\p)$ is species richness (the number of species present),
$\dvv{1}{\I}(\p)$ is the exponential of Shannon entropy, $\dvv{2}{\I}(\p)$
is the Gini--Simpson index (the reciprocal of the probability that two
randomly-chosen individuals are of the same species), and
$\dvv{\infty}{\I}(\p) = 1/\max_i p_i$ is the Berger--Parker index (a
measure of the dominance of the most abundant species).

Now allowing a general $\Z$, the diversity of order $2$ is $1 \big/
\sum_{i, j} p_i Z_{ij} p_j$.  Thus, diversity of order $2$ is the
reciprocal of the expected similarity between a random pair of individuals.
(The meaning given to `similarity' will determine the meaning of the
diversity measure: taking the coefficients $Z_{ij}$ to be genetic
similarities produces a genetic notion of diversity, and similarly
phylogenetic, taxonomic, and so on.)  Up to an increasing, invertible
transformation, this is the well-studied quantity known as Rao's quadratic
entropy.

Given distributions $\p$ and $\p'$ on the same list of species, different
values of $q$ may make different judgements on which of $\p$ and $\p'$ is
the more diverse.  For instance, with $\Z = \I$ and the two distributions
shown in Figure~\ref{fig:birds}, taking $q = 0$ 
makes community~(a) more diverse and embodies the first `extreme viewpoint'
described in Section~\ref{sec:spectrum}, whereas $q = \infty$ makes~(b)
more diverse and embodies the opposite extreme.

It is therefore most informative if we calculate the diversity of
\emph{all} orders $q \in [0, \infty]$.  The graph of $\dqz(\p)$ against $q$
is called the \demph{diversity profile} of $\p$.  Two distributions $\p$
and $\p'$ can be compared by plotting their diversity profiles on the same
axes.  If one curve is wholly above the other then the corresponding
distribution is unambiguously more diverse.  If they cross (as in
Figure~\ref{fig:main}(b), p.\pageref{fig:main}), then the judgement as
to which is the more diverse depends on how much importance is attached to
rare species.

The formula for $\dqz(\p)$ can be understood as follows.  

First, for a given species $i$, the quantity $(\Z\p)_i = \sum_j Z_{ij} p_j$
is the expected similarity between species $i$ and an individual chosen at
random.  Differently put, $(\Z\p)_i$ measures the ordinariness of the $i$th
species witin the community; in~\cite{MDISS}, it is called the `relative
abundance of species similar to the $i$th'.  Hence the mean ordinariness of
an individual in the community is $\sum_i p_i (\Z\p)_i$.  This measures the
\emph{lack} of diversity of the community, so its reciprocal is a measure
of diversity.  This is exactly $\dvv{2}{\Z}(\p)$.

To explain the diversity of orders $q \neq 2$, we recall the classical
notion of power mean.  Let $\p = (p_1, \ldots, p_n)$ be a finite
probability distribution and let $\x = (x_1, \ldots, x_n) \in [0,
  \infty)^n$, with $x_i > 0$ whenever $p_i > 0$.  For real $t \neq 0$, the
  \demph{power mean} of $\x$ of order $t$, weighted by $\p$, is
\[
M_t(\p, \x) 
=
\biggl( \sum_{i \in \supp(\p)} p_i x_i^t \biggr)^{1/t}
\]
\cite[Chapter~II]{HLP}.  This definition is extended to $t = 0$ and $t
= \pm \infty$ by taking limits in $t$, which gives
\[
M_{-\infty} (\p, \x)
=
\min_{i \in \supp(\p)} x_i,
\quad
M_0(\p, \x) 
=
\prod_{i \in \supp(\p)} \!\! x_i^{p_i},
\quad
M_\infty (\p, \x)
=
\max_{i \in \supp(\p)} x_i.
\]
Now, when we take the `mean ordinariness' in the previous paragraph, we can
replace the ordinary arithmetic mean (the case $t = 1$) by the power mean
of order $t = q - 1$.  Again taking the reciprocal, we obtain
exactly~\eqref{eq:the-measures}.  That is, 
\begin{equation}
\label{eq:mean-div}
\dqz(\p) = 1/M_{q - 1}(\p, \Z\p)
\end{equation}
for all $\p$, $\Z$, and $q \in [0, \infty]$.  So in all cases, diversity is
the reciprocal mean ordinariness of an individual within the community,
for varying interpretations of `mean'.

The diversity measures $\dqz(\p)$ have many good properties, discussed
in~\cite{MDISS}.  Crucially, they are \demph{effective numbers}: that is,
\[
\dvv{q}{\I}(1/n, \ldots, 1/n) = n
\]
for all $q$ and $n$.  This gives meaning to the quantities $\dqz(\p)$: if
$\dqz(\p) = 32.8$, say, then the community is nearly as diverse as a
community of $33$ completely dissimilar species in equal proportions.  With
the stronger assumptions~\eqref{eq:stronger-sim-hyps} on $\Z$, the value of
$\dqz(\p)$ always lies between $1$ and $n$.

Diversity profiles are decreasing: as less emphasis is given to rare
species, perceived diversity drops.  More precisely:
\begin{propn}
\label{propn:div-dichotomy}
Let $\p$ be a probability distribution on $\{1, \ldots, n\}$ and let $\Z$
be an $n \times n$ matrix satisfying~\eqref{eq:weaker-sim-hyps}.  If
$(\Z\p)_i$ has the same value $K$ for all $i \in \supp(\p)$ then $\dqz(\p)
= 1/K$ for all $q \in [0, \infty]$.  Otherwise, $\dqz(\p)$ is strictly
decreasing in $q \in [0, \infty]$.
\end{propn}

\begin{proof}
This is immediate from equation~\eqref{eq:mean-div} and a classical
result on power means~\cite[Theorem~16]{HLP}: $M_t(\p, \x)$ is increasing
in $t$, strictly so unless $x_i$ has the same value $K$ for all $i \in
\supp(\p)$, in which case it has constant value $K$.  
\end{proof}

So, any diversity profile is either constant or strictly decreasing.  The
first part of the next lemma states that diversity profiles are also
continuous:

\begin{lemma}
\label{lemma:div-cts}
Fix an $n \times n$ matrix $\Z$ satisfying~\eqref{eq:weaker-sim-hyps}.
Then:
\begin{enumerate}
\item 
\label{part:cts-q}
$\dqz(\p)$ is continuous in $q \in [0, \infty]$ for each distribution
$\p$; 

\item 
\label{part:cts-p}
$\dqz(\p)$ is continuous in $\p$ for each $q \in (0, \infty)$.
\end{enumerate}
\end{lemma}

\begin{proof}
See Propositions~A1 and~A2 of the appendix of~\cite{MDISS}.
\end{proof}

Finally, the measures have the sensible property that if some species have
zero abundance, then the diversity is the same as if they were not
mentioned at all.  To express this, we introduce some notation: given a
subset $B \sub \{1, \ldots, n\}$, we denote by $\Z_B$ the submatrix
$(Z_{ij})_{i, j \in B}$ of $\Z$.

\begin{lemma}[Absent species]
\label{lemma:abs}
Let $\Z$ be an $n \times n$ matrix satisfying~\eqref{eq:weaker-sim-hyps}.
Let $B \sub \{1, \ldots, n\}$, and let $\p$ be a probability distribution
on $\{1, \ldots, n\}$ such that $p_i = 0$ for all $i \not\in B$.  Then,
writing $\p'$ for the restriction of $\p$ to $B$,
\[
\dvv{q}{\Z_B}(\p')
= 
\dqz(\p) 
\]
for all $q \in [0, \infty]$.
\end{lemma}

\begin{proof}
This is trivial, and is also an instance of a more general naturality
property (Lemma~A13 in the appendix of~\cite{MDISS}).
\end{proof}

\section{Preparatory lemmas}
\label{sec:invt}

\textbf{For the rest of this work, fix an integer $n \geq 1$ and an $n
  \times n$ symmetric matrix $\Z$ of nonnegative reals whose diagonal
  entries are positive} (that is, \emph{strictly} greater than
zero).  Also write
\[
\Delta_n 
=
\bigl\{ (p_1, \ldots, p_n) \in \R^n 
\such 
p_i \geq 0, \ p_1 + \cdots + p_n = 1
\bigr\}
\]
for the set of probability distributions on $\{1, \ldots, n\}$.

To prove the main theorem, we begin by making two apparent digressions.

Let $\M$ be any matrix.  A \demph{weighting} on $\M$ is a column vector
$\w$ such that $\M\w$ is the column vector whose every entry is $1$.  It is
trivial to check that if both $\M$ and its transpose have at least one
weighting, then the quantity $\sum_i w_i$ is independent of the choice of
weighting $\w$ on $\M$; this quantity is called the \demph{magnitude}
$\mg{\M}$ of $\M$ \cite[Section~1.1]{MMS}.

When $\M$ is symmetric (the case of interest here), $\mg{\M}$ is defined
just as long as $\M$ has at least one weighting.  When $\M$ is invertible,
$\M$ has exactly one weighting and $\mg{\M}$ is the sum of all the entries
of $\M^{-1}$.

The second digression concerns the dichotomy expressed in
Proposition~\ref{propn:div-dichotomy}: every diversity profile is either
constant or strictly decreasing.  We now ask: which distributions have
constant diversity profile?

This question turns out to have a clean answer in terms of weightings and
magnitude.  To state it, we make some further definitions.

\begin{defn}
A probability distribution $\p$ on $\{1, \ldots, n\}$ is \demph{invariant}
if $\dqz(\p) = \dvv{q'}{\Z}(\p)$ for all $q, q' \in [0, \infty]$.
\end{defn}

Let $B \sub \{1, \ldots, n\}$, and let $\vecstyle{0} \neq \w \in [0,
  \infty)^B$ be a nonnegative vector.  Then there is a probability
  distribution $\pext{\w}$ on $\{1, \ldots, n\}$ defined by
\[
(\pext{\w})_i
=
\begin{cases}
w_i/\sum_{j \in B} w_j  &\text{if } i \in B     \\
0                       &\text{otherwise.}
\end{cases}
\]
In particular, let $B$ be a nonempty subset of $\{1, \ldots, n\}$ and $\w$
a nonnegative weighting on $\Z_B = (Z_{ij})_{i, j \in B}$.  Then $\w \neq
\vecstyle{0}$, so $\pext{\w}$ is defined, and $\pext{\w}_i = w_i/\mg{\Z_B}$
for all $i \in B$.

\begin{lemma}
\label{lemma:invt}
The following are equivalent for $\p \in \Delta_n$:
%
\begin{enumerate}
\item 
\label{part:invt-invt}
$\p$ is invariant;

\item
\label{part:invt-supp}
$(\Z\p)_i = (\Z\p)_j$ for all $i, j \in \supp(\p)$;

\item
\label{part:invt-ext}
$\p = \pext{\w}$ for some nonnegative weighting $\w$ on $\Z_B$ and some
nonempty subset $B \sub \{1, \ldots, n\}$.
\end{enumerate}
Moreover, in the situation of~\bref{part:invt-ext}, $\dqz(\p) = \mg{\Z_B}$
for all $q \in [0, \infty]$.
\end{lemma}

\begin{proof}
\bref{part:invt-invt}$\iff$\bref{part:invt-supp} is immediate
from Proposition~\ref{propn:div-dichotomy}.

For \bref{part:invt-supp}$\implies$\bref{part:invt-ext},
assume~\bref{part:invt-supp}.  Put $B = \supp(\p)$ and write $K = (\Z\p)_i$
for any $i \in B$.  Then $K > 0$, so we may define $\w \in \R^B$ by $w_i =
p_i/K$ \ ($i \in B$).  Evidently $\p = \pext{\w}$ and $\w$ is nonnegative.
Furthermore, $\w$ is a weighting on $\Z_B$, since whenever $i \in B$,
\[
(\Z_B \w)_i     
=
\sum_{j \in B} Z_{ij} p_j/K
=
\sum_{j = 1}^n Z_{ij} p_j/K
=
1.
\]

Finally, for \bref{part:invt-ext}$\implies$\bref{part:invt-supp} and
`moreover', take $B$ and $\w$ as in~\bref{part:invt-ext}.  Then
$\supp(\pext{\w}) \sub B$, so for all $i \in \supp(\pext{\w})$,
\[
\bigl(\Z\cdot\pext{\w}\bigr)_i 
=
\bigl(\Z_B \w/\mg{\Z_B}\bigr)_i
=
1/\mg{\Z_B}.
\]
Hence $\dqz(\pext{\w}) = \mg{\Z_B}$ for all $q \in [0, \infty]$ by
Proposition~\ref{propn:div-dichotomy}. 
\end{proof}

We now prove a result that is much weaker than the main theorem, but
will act as a stepping stone in the proof.

\begin{lemma}  
\label{lemma:01}
For each $q \in (0, 1)$, there exists an invariant distribution that
maximizes $\dqz$.
\end{lemma}

\begin{proof}
Let $q \in (0, 1)$.  Then $\dqz$ is continuous on the compact space
$\Delta_n$ (Lemma~\ref{lemma:div-cts}\bref{part:cts-p}), so attains a
maximum at some point $\p$.  Take $j, k \in \supp(\p)$ such that $(\Z\p)_j$
is least and $(\Z\p)_k$ is greatest.  By Lemma~\ref{lemma:invt}, it is
enough to prove that $(\Z\p)_j = (\Z\p)_k$.

Define $\delta_j \in \Delta_n$ by taking $(\delta_j)_i$ to be the Kronecker
delta $\delta_{ji}$, and $\delta_k$ similarly.  Then $\p + t(\delta_j -
\delta_k) \in \Delta_n$ for all real $t$ sufficiently close to $0$, and
\begin{align}
0       &
=
\frac{d}{dt} 
\Bigl( 
\dqz \bigl(\p + t(\delta_j - \delta_k)\bigr)^{1 - q} 
\Bigr)\bigg|_{t = 0}
\label{eq:path-critical}        \\
&
=
(q - 1) \Biggl( 
\sum_{i \in \supp(\p)} Z_{ij} p_i (\Z\p)_i^{q - 2} -
\sum_{i \in \supp(\p)} Z_{ik} p_i (\Z\p)_i^{q - 2} 
\Biggr)
+ (\Z\p)_j^{q - 1} - (\Z\p)_k^{q - 1}
\label{eq:path-explicit}        \\
&
\geq
(q - 1) \Biggl( 
\sum_{i = 1}^n Z_{ij} p_i (\Z\p)_j^{q - 2} -
\sum_{i = 1}^n Z_{ik} p_i (\Z\p)_k^{q - 2} 
\Biggr)
+ (\Z\p)_j^{q - 1} - (\Z\p)_k^{q - 1}
\label{eq:path-ineq}            \\
&
=
q \bigl( (\Z\p)_j^{q - 1} - (\Z\p)_k^{q - 1} \bigr) 
\label{eq:path-sym}       \\
&
\geq 0,
\label{eq:path-coup}
\end{align}
where~\eqref{eq:path-critical} holds because $\p$ is a supremum,
\eqref{eq:path-explicit}~is a routine computation, \eqref{eq:path-ineq}~and
\eqref{eq:path-coup}~follow from the defining properties of $j$ and $k$,
and~\eqref{eq:path-sym} uses the symmetry of $\Z$.  Equality therefore
holds throughout, and in particular in~\eqref{eq:path-coup}.  Hence
$(\Z\p)_j = (\Z\p)_k$, as required.
\end{proof}

An alternative proof uses Lagrange multipliers, but is complicated by the
possibility that $\dqz$ attains its maximum on the boundary of $\Delta_n$.

The result we have just proved only concerns the maximization of $\dqz$ for
specific values of $q$, but the following lemma will allow us to deduce
results about maximization for all $q$ simultaneously.

\begin{defn}
A probability distribution on $\{1, \ldots, n\}$ is \demph{maximizing}
if it maximizes $\dqz$ for all $q \in [0, \infty]$.
\end{defn}

\begin{lemma}
\label{lemma:invt-varying-q}
For $0 \leq q' \leq q \leq \infty$, any invariant distribution that
maximizes $\dvv{q'}{\Z}$ also maximizes $\dqz$.  In particular, any
invariant distribution that maximizes $\dvv{0}{\Z}$ is maximizing.
\end{lemma}

\begin{proof}
Let $0 \leq q' \leq q \leq \infty$ and let $\p$ be an invariant
distribution that maximizes $\dvv{q'}{\Z}$.  Then for all $\rvec \in
\Delta_n$,
\[
\dqz(\rvec) 
\leq 
\dvv{q'}{\Z}(\rvec)
\leq
\dvv{q'}{\Z}(\p)
=
\dqz(\p),
\]
since diversity profiles are decreasing
(Proposition~\ref{propn:div-dichotomy}).
\end{proof}

\section{The main theorem}
\label{sec:main}

\begin{thm}[Main theorem]
\label{thm:main}
There exists a probability distribution on $\{1, \ldots, n\}$ that
maximizes $\dqz$ for all $q \in [0, \infty]$.  Moreover, the maximum
diversity $\sup_{\p \in \Delta_n} \dqz(\p)$ is independent of $q \in [0,
  \infty]$.
\end{thm}

\begin{proof}
An equivalent statement is that there exists an invariant maximizing
distribution.  To prove this, choose a decreasing sequence
$(q_\lambda)_{\lambda = 1}^\infty$ in $(0, 1)$ converging to $0$.  By
Lemma~\ref{lemma:01}, we can choose for each $\lambda \geq 1$ an invariant
distribution $\p^\lambda$ that maximizes $\dvv{q_\lambda}{\Z}$.  Since
$\Delta_n$ is compact, we may assume (by passing to a subsequence if
necessary) that the sequence $(\p^\lambda)$ converges to some point $\p \in
\Delta_n$.  We will show that $\p$ is invariant and maximizing.

We show that $\p$ is invariant using Lemma~\ref{lemma:invt}.  Let $i, j \in
\supp(\p)$.  Then $i, j \in \supp(\p^\lambda)$ for all $\lambda \gg 0$, so
$(\Z\p^\lambda)_i = (\Z\p^\lambda)_j$ for all $\lambda \gg 0$, and letting
$\lambda \to \infty$ gives $(\Z\p)_i = (\Z\p)_j$.

To show that $\p$ is maximizing, first note that $\p^{\lambda'}$ maximizes
$\dvv{q_\lambda}{\Z}$ whenever $\lambda' \geq \lambda \geq 1$ (by
Lemma~\ref{lemma:invt-varying-q}).  Fixing $\lambda$ and letting $\lambda'
\to \infty$, this implies that $\p$ maximizes $\dvv{q_\lambda}{\Z}$, since
$\dvv{q_\lambda}{\Z}$ is continuous
(Lemma~\ref{lemma:div-cts}\bref{part:cts-p}).

Thus, $\p$ maximizes $\dvv{q_\lambda}{\Z}$ for all $\lambda$.  But
$q_\lambda \to 0$ as $\lambda \to \infty$, and diversity is continuous in
its order (Lemma~\ref{lemma:div-cts}\bref{part:cts-q}), so $\p$ maximizes
$\dvv{0}{\Z}$.  Since $\p$ is invariant, Lemma~\ref{lemma:invt-varying-q}
implies that $\p$ is maximizing.
\end{proof}

The theorem can be understood as follows (Figure~\ref{fig:main}(a)).  Each
particular value of the viewpoint parameter $q$ ranks the set of all
distributions $\p$ in order of diversity, with $\p$ placed above $\p'$ when
$\dqz(\p) > \dqz(\p')$.  Different values of $q$ rank the set of
distributions differently.  Nevertheless, there is a distribution $\pmax$
that is at the top of every ranking.  This is the content of the first half
of Theorem~\ref{thm:main}.

Alternatively, we can visualize the theorem in terms of diversity profiles
(Figure~\ref{fig:main}(b)).  Diversity profiles may cross, reflecting the
different priorities embodied by different values of $q$.  But there is at
least one distribution $\pmax$ whose profile is above every other profile;
moreover, its profile is constant.

\begin{figure}
\begin{center}
\setlength{\unitlength}{1mm}
\begin{picture}(50,37)
\cell{25}{-2}{t}{(a)}
\cell{10}{0}{b}{$q = 0$}
\cell{10}{5}{b}{\includegraphics[width=20\unitlength]{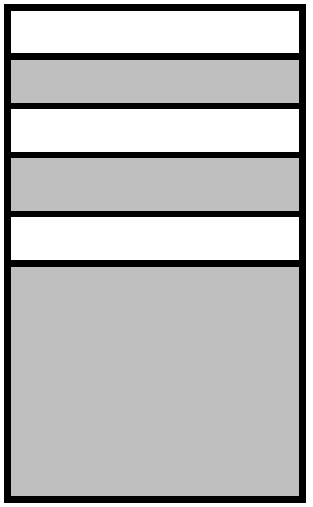}}
\cell{10}{35.2}{c}{\pmaxcolor$\pmax$}
\cell{10}{28.9}{c}{\pcolor$\p$}
\cell{10}{22.0}{c}{\pprimecolor$\p\scriptstyle'$}
\cell{40}{0}{b}{$q = 2$}
\cell{40}{5}{b}{\includegraphics[width=20\unitlength]{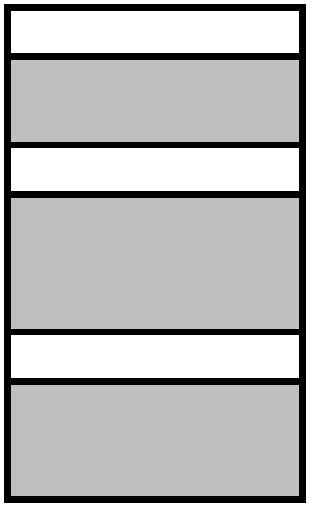}}
\cell{40}{35.2}{c}{\pmaxcolor$\pmax$}
\cell{40}{26.5}{c}{\pprimecolor$\p\scriptstyle'$}
\cell{40}{14.5}{c}{\pcolor$\p$}
\end{picture}%
\hspace*{10\unitlength}%
\begin{picture}(60,37)
\cell{30}{-2}{t}{(b)}
\cell{30}{5}{b}{\includegraphics[width=60\unitlength]{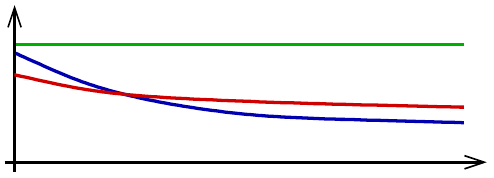}}
\cell{1}{15}{r}{$\dqz$}
\cell{30}{5}{t}{$q$}
\cell{10}{18.5}{c}{\pcolor$\p$}
\cell{6}{14}{c}{\pprimecolor$\p\scriptstyle'$}
\cell{16}{23}{c}{\pmaxcolor$\pmax$}
\end{picture}
\end{center}
\caption{Visualizations of the main theorem: (a)~in terms of how different
  values of $q$ rank the set of distributions, and (b)~in terms of
  diversity profiles.}
\label{fig:main}
\end{figure}

Theorem~\ref{thm:main} immediately implies:
\begin{cor}
\label{cor:max-invt}
Every maximizing distribution is invariant.
\qed
\end{cor}

This result can be partially understood as follows.  For Shannon entropy, and
more generally any of the R\'enyi entropies, the maximizing distribution is
obtained by taking the relative abundance $p_i$ to be the same for all
species $i$.  This is no longer true when inter-species similarities are
taken into account.  However, what is approximately true is that diversity
is maximized when $(\Z\p)_i$, the relative abundance of species
\emph{similar} to the $i$th, is the same for all species $i$.  This follows
from Corollary~\ref{cor:max-invt} together with the characterization of
invariant distributions in Lemma~\ref{lemma:invt}\bref{part:invt-supp}; but
it is only `approximately true' because it is only guaranteed that
$(\Z\p)_i = (\Z\p)_j$ when $i$ and $j$ both belong to the support of $\p$,
not for all $i$ and $j$.  It may in fact be that some or all maximizing
distributions do not have full support, a phenomenon we examine in
Section~\ref{sec:elim}.

The second half of Theorem~\ref{thm:main} tells us that associated with the
matrix $\Z$ is a numerical invariant, the constant value of a maximizing
distribution:
\begin{defn}
The \demph{maximum diversity} of $\Z$ is $\Dmax{\Z} = \sup_{\p \in
  \Delta_n} \dqz(\p)$, for any $q \in [0, \infty]$.
\end{defn}
We show how to compute $\Dmax{\Z}$ in the next section.

If a distribution $\p$ maximizes diversity of order $2$, say, must it also
maximize diversity of orders $1$ and $\infty$?  The answer turns out to be
yes:

\begin{cor}
\label{cor:irrelevance}
Let $\p$ be a probability distribution on $\{1, \ldots, n\}$.  If $\p$
maximizes $\dqz$ for some $q \in (0, \infty]$ then $\p$ maximizes $\dqz$
for all $q \in [0, \infty]$.
\end{cor}

\begin{proof}
Let $q \in (0, \infty]$ and let $\p$ be a distribution maximizing $\dqz$.
Then 
\[
\dqz(\p)
\leq
\dvv{0}{\Z}(\p)
\leq
\Dmax{\Z}
=
\dqz(\p),
\]
where the first inequality holds because diversity profiles are decreasing.
So equality holds throughout.  Now $\dqz(\p) = \dvv{0}{\Z}(\p)$ with $q
\neq 0$, so Proposition~\ref{propn:div-dichotomy} implies that $\p$ is
invariant.  But also $\dvv{0}{\Z}(\p) = \Dmax{\Z}$, so $\p$ maximizes
$\dvv{0}{\Z}$.  Hence by Lemma~\ref{lemma:invt-varying-q}, $\p$ is
maximizing.
\end{proof}

The significance of this corollary is that if we wish to find a
distribution that maximizes diversity of all orders $q$, it suffices to
find a distribution that maximizes diversity of a single nonzero order.

The hypothesis that $q > 0$ in Corollary~\ref{cor:irrelevance} cannot be
dropped.  Indeed, take $\Z = \I$.  Then $\dvv{0}{\I}(\p)$ is species
richness (the cardinality of $\supp(\p)$), which is maximized by any
distribution $\p$ of full support, whereas $\dvv{1}{\I}(\p)$ is the
exponential of Shannon entropy, which is maximized only when $\p$ is
uniform.

\section{The computation theorem}
\label{sec:comp}

The main theorem guarantees the existence of a maximizing distribution
$\pmax$, but does not tell us how to find it.  It also states that 
$\dqz(\pmax)$ is independent of $q$, but does not tell us what its value
is.  The following result repairs both omissions.

\begin{thm}[Computation theorem]
\label{thm:comp}
\begin{enumerate}
\item 
\label{part:comp-value}
For all $q \in [0, \infty]$,
\begin{equation}
\label{eq:sup-max}
\sup_{\p \in \Delta_n} \dqz(\p)
=
\max_B \mg{\Z_B}
\end{equation}
where the maximum is over all $B \sub \{1, \ldots, n\}$ such that $\Z_B$
admits a nonnegative weighting.

\item
\label{part:comp-dist}
The maximizing distributions are precisely those of the form $\pext{\w}$
where $\w$ is a nonnegative weighting on $\Z_B$ for some $B$ attaining the
maximum in~\eqref{eq:sup-max}.
\end{enumerate}
\end{thm}

\begin{proof}
Let $q \in [0, \infty]$.  Then
\begin{align}
&
\sup \{ \dqz(\p) \such \p \in \Delta_n \} 
\nonumber
\\
&
=
\sup \{ \dqz(\p) \such \p \in \Delta_n, \ \p \text{ is invariant} \}    
\label{eq:comp-invt}    \\
&
=
\sup \{ \mg{\Z_B} \such \emptyset \neq B \sub \{1, \ldots, n\},
\ \Z_B \text{ admits a nonnegative weighting} \}    
\label{eq:comp-wtg}     \\
&
=
\max \{ \mg{\Z_B} \such B \sub \{1, \ldots, n\},
\ \Z_B \text{ admits a nonnegative weighting} \},
\label{eq:comp-empty}     
\end{align}
where~\eqref{eq:comp-invt} follows from the fact that there is an invariant
maximizing distribution (Theorem~\ref{thm:main}),
\eqref{eq:comp-wtg}~follows from Lemma~\ref{lemma:invt},
and~\eqref{eq:comp-empty} follows from the trivial fact that
$\mg{\Z_B} \geq 0 = \mg{\Z_\emptyset}$ whenever $\Z_B$ admits a nonnegative
weighting. 

This proves~\bref{part:comp-value}.  Every maximizing distribution is
invariant (Corollary~\ref{cor:max-invt}), so~\bref{part:comp-dist} follows
from Lemma~\ref{lemma:invt}.
\end{proof}

\begin{remark}
\label{rmk:algorithm}
The theorem provides a finite-time algorithm for finding all the maximizing
distributions and computing $\Dmax{\Z}$, as follows.  For each of the $2^n$
subsets $B$ of $\{1, \ldots, n\}$, perform some simple linear algebra to
find the space of nonnegative weightings on $\Z_B$; if this space is
nonempty, call $B$ \demph{feasible} and record the magnitude $\mg{\Z_B}$.
Then $\Dmax{\Z}$ is the maximum of all the recorded magnitudes.  For each
feasible $B$ such that $\mg{\Z_B} = \Dmax{\Z}$, and each nonnegative
weighting $\w$ on $\Z_B$, the distribution $\pext{\w}$ is maximizing.  This
generates all of the maximizing distributions.

This algorithm takes exponentially many steps in $n$, and
Remark~\ref{rmk:no-quick-clique} provides strong evidence that the time
taken cannot be reduced to a polynomial in $n$.  But the situation is not
as hopeless as it might appear, for two reasons.

First, each step of the algorithm is fast, consisting as it does of solving
a system of linear equations.  For instance, in an implementation in
\textsc{Matlab} on a standard laptop, with no attempt at optimization, the
maximizing distributions of $25 \times 25$ matrices were computed in a few
seconds.%
\footnote{We thank Christina Cobbold for carrying out this implementation.}
Second, for certain classes of matrices $\Z$, we can make substantial
improvements in computing time, as observed in
Section~\ref{sec:examples-pos-def}.
\end{remark}

\section{Simple examples}
\label{sec:examples-simple}

The next three sections give examples of the main results, beginning here
with some simple, specific examples.  

\begin{example}
\label{eg:naive}
First consider the naive model $\Z = \I$, in which different species are
deemed to be entirely dissimilar.  As noted in Section~\ref{sec:measures},
$\dvv{q}{\I}(\p)$ is the exponential of the R\'enyi entropy of order $q$.
It is well-known that R\'enyi entropy of any order $q > 0$ is maximized
uniquely by the uniform distribution.  This result also follows trivially
from Corollary~\ref{cor:irrelevance}: for clearly $\dvv{\infty}{\I}(\p) =
1/\max_i p_i$ is uniquely maximized by the uniform distribution, and the
corollary implies that the same is true for all values of $q > 0$.
Moreover, $\Dmax{\I} = \mg{\I} = n$.
\end{example}

\begin{example}
\label{eg:three}
For a general matrix $\Z$ satisfying~\eqref{eq:stronger-sim-hyps}, a
two-species system is always maximized by the uniform distribution $p_1 =
p_2 = 1/2$.  When $n = 3$, however, nontrivial examples arise.  For
instance, take the system shown in Figure~\ref{fig:three}, consisting of
one species of newt and two species of frog.  Let us first consider
intuitively what we expect the maximizing distribution to be, then
compare this with the answer given by Theorem~\ref{thm:comp}.

\begin{figure}
\begin{center}
\setlength{\unitlength}{1mm}
\begin{picture}(120,16)(0,-1)
\cell{0}{7}{l}{%
$\Z 
= 
\begin{pmatrix} 
1       &0.4    &0.4    \\
0.4     &1      &0.9    \\
0.4     &0.9    &1
\end{pmatrix}$}
\thicklines
\put(41,-1.5){\includegraphics[width=80\unitlength]{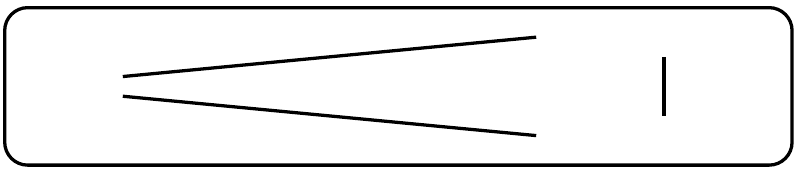}}
\cell{48}{7}{c}{newt}
\cell{117}{12}{r}{frog species A}
\cell{117}{2}{r}{frog species B}
\cell{108.5}{7}{l}{0.9}
\cell{73}{12}{c}{0.4}
\cell{73}{2}{c}{0.4}
\end{picture}%
\end{center}
\caption{Hypothetical three-species system.  Distances between species
  indicate degrees of dissimilarity between them (not to scale).}
\label{fig:three}
\end{figure}

If we ignore the fact that the two frog species are more similar to each
other than they are to the newt, then (as in Example~\ref{eg:naive}) the
maximizing distribution is $(1/3, 1/3, 1/3)$.  At the other extreme, if we
regard the two frog species as essentially identical then effectively there
are only two species, newts and frogs, so the maximizing distribution gives
relative abundance $0.5$ to the newt and $0.5$ to the frogs.  So with this
assumption, we expect diversity to be maximized by the distribution $(0.5,
0.25, 0.25)$.

Intuitively, then, the maximizing distribution should lie between these two
extremes.  And indeed, it does: implementing the algorithm in
Remark~\ref{rmk:algorithm} (or using Proposition~\ref{propn:pos-def-max}
below) reveals that the unique maximizing distribution is $(0.478, 0.261,
0.261)$.
%
%
\end{example}

One of our standing hypotheses on $\Z$ is symmetry.  The last of our simple
examples shows that if $\Z$ is no longer assumed to be symmetric, then the
main theorem fails in every respect.

\begin{example}
\label{eg:nonsym}
Let $\Z = \bigl( \begin{smallmatrix} 1 & 1/2 \\ 0 & 1 \end{smallmatrix}
\bigr)$, which satisfies all of our standing hypotheses except symmetry.  
Consider a distribution $\p = (p_1, p_2) \in \Delta_2$.  If $\p$ is $(1,
0)$ or $(0, 1)$ then $\dqz(\p) = 1$ for all $q$.  Otherwise,
\begin{align}
\dvv{0}{\Z}(\p) &
=
3 - \frac{2}{1 + p_1},  
\label{eq:nonsym-0}
\\
\dvv{2}{\Z}(\p) &
=
\frac{2}{3(p_1 - 1/2)^2 + 5/4}, 
\label{eq:nonsym-2}
\\
\dvv{\infty}{\Z}(\p) &
=
\begin{cases}
1/(1 - p_1)     &
\text{if } p_1 \leq 1/3,        \\
2/(1 + p_1)     &
\text{if } p_1 \geq 1/3.
\end{cases}
\label{eq:nonsym-infty}
\end{align}
From~\eqref{eq:nonsym-0} it follows that $\sup_{\p \in \Delta_2}
\dvv{0}{\Z}(\p) = 2$.  However, this supremum is not attained;
$\dvv{0}{\Z}(\p) \to 2$ as $\p \to (1, 0)$, but $\dvv{0}{\Z}(1, 0) = 1$.
Equations~\eqref{eq:nonsym-2} and~\eqref{eq:nonsym-infty} imply that
\[
\sup_{\p \in \Delta_2} \dvv{2}{\Z}(\p) = 1.6,
\qquad
\sup_{\p \in \Delta_2} \dvv{\infty}{\Z}(\p) = 1.5,
\]
with unique maximizing distributions $(1/2, 0)$ and $(1/3, 2/3)$
respectively.

Thus, when $\Z$ is not symmetric, the main theorem fails comprehensively:
the supremum $\sup_{\p \in \Delta_n} \dvv{0}{\Z}(\p)$ may not be attained;
there may be no distribution maximizing $\sup_{\p \in \Delta_n}
\dvv{q}{\Z}(\p)$ for all $q$ simultaneously; and that supremum may vary
with $q$.
\end{example}

Perhaps surprisingly, nonsymmetric similarity matrices $\Z$ do have
practical uses.  For example, it is shown in \cite[Proposition~A7]{MDISS}
that the mean phylogenetic diversity measures of Chao, Chiu and
Jost~\cite{CCJ} are a special case of the measures $\dqz(\p)$, obtained by
taking a particular $\Z$ depending on the phylogenetic tree concerned.
This $\Z$ is usually nonsymmetric, reflecting the asymmetry of evolutionary
time.  More generally, the case for dropping the symmetry axiom for metric
spaces was made in~\cite{LawvMSG}, and Gromov has argued that symmetry
`unpleasantly limits many applications' \cite[p.xv]{GromMSR}.  So the fact
that our maximization theorem fails for nonsymmetric $\Z$ is an important
restriction.

\section{Maximum diversity on graphs}
\label{sec:examples-graphs}

Consider those matrices $\Z$ for which each similarity coefficient $Z_{ij}$
is either $0$ or $1$.  A matrix $\Z$ of this form amounts to a (finite,
undirected) reflexive graph with vertex-set $\{1, \ldots, n\}$, with an
edge between $i$ and $j$ if and only if $Z_{ij} = 1$.  (That is, $\Z$ is
the \demph{adjacency matrix} of the graph.)  Our standing hypotheses on
$\Z$ then imply that $Z_{ii} = 1$ for all $i$, so every vertex has a loop
on it; this is the meaning of \demph{reflexive}.

What is the maximum diversity of the adjacency matrix of a graph?  

To state the answer, we recall some terminology.  Vertices $x$ and $y$ of a
graph are \demph{adjacent}, written $x \adjt y$, if there is an edge
between them.  (In particular, every vertex of a reflexive graph is
adjacent to itself.)  A set of vertices is \demph{independent} if no two
distinct vertices are adjacent.  The \demph{independence number}
$\alpha(G)$ of a graph $G$ is the maximal cardinality of an independent set
of vertices of $G$.

\begin{propn}
\label{propn:graph-main}
Let $G$ be a reflexive graph with adjacency matrix $\Z$.  Then the maximum
diversity $\Dmax{\Z}$ is equal to the independence number $\alpha(G)$.
\end{propn}

\begin{proof}
We will maximize the diversity of order $\infty$ and apply
Theorem~\ref{thm:main}.  For any probability distribution $\p$ on the
vertex-set $\{1, \ldots, n\}$, we have 
\begin{equation}
\label{eq:graph-infty}
\dvv{\infty}{\Z}(\p)
=
1 \Bigl/
\max_{i \in \supp(\p)} \sum_{j \csuch \adjto{i}{j}} p_j.
\end{equation}

First we show that $\Dmax{\Z} \geq \alpha(G)$.  Choose an independent set
$B$ of maximal cardinality, and define $\p \in \Delta_n$ by
\[
p_i
=
\begin{cases}
1/\alpha(G)     &\text{if } i \in B,    \\
0               &\text{otherwise}.
\end{cases}
\]
For each $i \in \supp(\p)$, the sum on the right-hand side of
equation~\eqref{eq:graph-infty} is $1/\alpha(G)$.  Hence
$\dvv{\infty}{\Z}(\p) = \alpha(G)$, and so $\alpha(G) \leq \Dmax{\Z}$.

Now we show that $\Dmax{\Z} \leq \alpha(G)$.  Let $\p \in \Delta_n$.
Choose an independent set $B \sub \supp(\p)$ with maximal cardinality among
all independent subsets of $\supp(\p)$.  Then every vertex of $\supp(\p)$
is adjacent to at least one vertex in $B$, otherwise we could adjoin it to
$B$ to make a larger independent subset.  Hence
\[
\sum_{\,i \in B\,} \sum_{j \csuch \adjto{i}{j}} p_j
=
\sum_{\,i \in B\,} \sum_{j \in \supp(\p) \csuch \adjto{i}{j}} p_j
\geq
\sum_{j \in \supp(\p)} p_j
=
1.
\]
So there exists $i \in B$ such that $\sum_{j \csuch \adjto{i}{j}} p_j \geq
1/\# B$, where $\# B$ denotes the cardinality of $B$.  But $\# B \leq
\alpha(G)$, so
\[
\max_{i \in \supp(\p)} \sum_{j \csuch \adjto{i}{j}} p_j 
\geq
1/\alpha(G),
\]
as required.
\end{proof}

\begin{remark}
\label{rmk:graph-proof}
The first part of the proof (together with Corollary~\ref{cor:irrelevance})
shows that a maximizing distribution can be constructed by taking the
uniform distribution on some independent set of largest cardinality, then
extending by zero to the whole vertex-set.  Except in the trivial case $\Z
= \I$, this maximizing distribution never has full support.  We return to
this point in Section~\ref{sec:elim}.
\end{remark}

\begin{example}
\label{eg:graph-lin3}
The reflexive graph $G =
\mbox{\ensuremath{\bullet\gedge\bullet\gedge\bullet}}$ (loops not shown) has
adjacency matrix $\Z = \Bigl( \begin{smallmatrix} 1 &1 &0 \\ 1& 1& 1 \\ 0&
  1& 1 \end{smallmatrix} \Bigr)$.  The independence number of $G$ is $2$;
this, then, is the maximum diversity of $\Z$.  There is a unique
independent set of cardinality $2$, and a unique maximizing distribution,
$(1/2, 0, 1/2)$.
\end{example}

\begin{example}
\label{eg:graph-lin4}
The reflexive graph
$\mbox{\ensuremath{\bullet\gedge\bullet\gedge\bullet\gedge\bullet}}$ again has
independence number $2$.  There are three independent sets of maximal
cardinality, so by Remark~\ref{rmk:graph-proof}, there are at least three
maximizing distributions,
\[
(1/2, 0, 1/2, 0),
\qquad
(1/2, 0, 0, 1/2),
\qquad
(0, 1/2, 0, 1/2),
\]
all with different supports.  (The possibility of multiple maximizing
distributions was also observed in the case $q = 2$ by Pavoine and
Bonsall~\cite{PaBo}.)
In fact, there are further maximizing distributions not constructed in the
proof of Proposition~\ref{propn:graph-main}, namely, $(1/2, 0, t, 1/2 - t)$
and $(1/2 - t, t, 0, 1/2)$ for any $t \in (0, 1/2)$.
\end{example}

\begin{example}
Let $d$ be a metric on $\{1, \dots, n\}$.  For a given $\epsln > 0$, the
\demph{covering number} $N(d, \epsln)$ is the minimum
cardinality of a subset $A \sub \{1, \dots, n\}$ such that
\[
\bigcup_{i \in A} B(i, \epsln) = \{1, \dots, n\},
\]
where $B(i,\epsln) = \{j \such d(i,j) \leq \epsln\}$.  The number $\log
N(d,\epsln)$ is known as the \demph{$\epsln$-entropy} of $d$~\cite{Kolm}.

Now define a matrix $\Z^\epsln$ by
\[
Z^\epsln_{ij} 
= 
\begin{cases} 
1 	& \text{if } d(i,j) \leq \epsln,	\\
0 	& \text{otherwise.} 
\end{cases}
\]
Then $\Z^\epsln$ is the adjacency matrix of the reflexive graph $G$ with
vertices $\{1, \dots, n\}$ and $i \sim j$ if and only if $d(i, j) \leq
\epsln$.  Thus, a subset of $B \sub \{1, \dots, n\}$ is independent in $G$
if and only if $d(i, j) > \epsln$ for every $i, j \in B$.  It is a
consequence of the triangle inequality that
\[
N(d, \epsln) \leq \alpha(G) \leq N(d, \epsln/2),
\]
and so by Proposition~\ref{propn:graph-main},
\[
N(d, \epsln) \leq \Dmax{\Z^\epsln} \leq N(d, \epsln/2).
\]
Recalling that $\log \dqz$ extends the classical notion of R\'enyi entropy,
this thoroughly justifies the name of $\epsln$-entropy (which was
originally justified by vague analogy).
\end{example}

The moral of the proof of Proposition~\ref{propn:graph-main} is that by
performing the simple task of maximizing diversity of order $\infty$, we
automatically maximize diversity of all other orders.  Here is an example
of how this can be exploited.

Recall that every graph $G$ has a \demph{complement} $\gcomp{G}$, with the
same vertex-set as $G$; two vertices are adjacent in $\gcomp{G}$ if and
only if they are not adjacent in $G$.  Thus, the complement of a reflexive
graph is irreflexive (has no loops), and vice versa.  A set $B$ of vertices
in an irreflexive graph $X$ is a \demph{clique} if all pairs of distinct
elements of $B$ are adjacent in $X$.  The \demph{clique number}
$\cliqueno{X}$ of $X$ is the maximal cardinality of a clique in $X$.  Thus,
$\cliqueno{X} = \alpha(\gcomp{X})$.

We now recover a result of Berarducci, Majer and Novaga
\cite[Proposition~5.10]{BMN}.

\begin{cor}
\label{cor:capacity}
Let $X$ be an irreflexive graph.  Then
\[
\sup_\p \sum_{(i, j) \csuch \adjto{i}{j}} p_i p_j
=
1 - \frac{1}{\cliqueno{X}}
\]
where the supremum is over probability distributions $\p$ on the vertex-set
of $X$, and the sum is over pairs of adjacent vertices of $X$.
\end{cor}

\begin{proof}
Write $\{1, \ldots, n\}$ for the vertex-set of $X$, and $\Z$ for the
adjacency matrix of the reflexive graph $\gcomp{X}$.  Then for all $\p \in
\Delta_n$, 
\begin{align*}
\sum_{(i, j) \csuch \adjtoin{i}{j}{X}} p_i p_j        &
=
\sum_{i, j = 1}^n p_i p_j - 
\sum_{(i, j) \csuch \adjtoin{i}{j}{\gcomp{X}}} p_i p_j  \\
&
=
1 - \sum_{i, j = 1}^n p_i Z_{ij} p_j    
=
1 - 1\big/\dvv{2}{\Z}(\p).
\end{align*}%
Hence by Theorem~\ref{thm:main} and Proposition~\ref{propn:graph-main},
\[
\sup_{\p \in \Delta_n} \sum_{(i, j) \csuch \adjtoin{i}{j}{X}} p_i p_j    
=
1 - \frac{1}{\Dmax{\p}}
= 
1 - \frac{1}{\alpha(\gcomp{X})}
=
1 - \frac{1}{\cliqueno{X}}.
\]
\end{proof}

It follows from this proof and Remark~\ref{rmk:graph-proof} that $\sum_{(i,
  j)\csuch \adjto{i}{j}} p_i p_j$ can be maximized as follows: take the
uniform distribution on some clique in $X$ of maximal cardinality, then
extend by zero to the whole vertex-set.

\begin{remark}
\label{rmk:no-quick-clique}
Proposition~\ref{propn:graph-main} implies that computationally, finding
the maximum diversity of an arbitrary $\Z$ is at least as hard as finding
the independence number of a reflexive graph.  This is a very well-studied
problem, usually presented in its dual form (find the clique number of an
irreflexive graph) and called the maximum clique problem~\cite{Karp}.  It
is $\mathbf{NP}$-hard, so on the assumption that $\mathbf{P} \neq
\mathbf{NP}$, there is no polynomial-time algorithm for computing maximum
diversity, nor even for computing the support of a maximizing distribution.
\end{remark}

\section{Positive definite similarity matrices}
\label{sec:examples-pos-def}

The theory of magnitude of metric spaces runs most smoothly when the
matrices $\Z$ concerned are positive definite~\cite{MeckPDM,MeckMDC}.  We
will see that positive (semi)definiteness is also an important condition
when maximizing diversity.

Any positive definite matrix is invertible and therefore has a unique
weighting.  (A positive semidefinite matrix need not have a weighting at
all.)  Now the crucial fact about magnitude is:

\begin{lemma}
\label{lemma:pos-def-variational}
Let $\M$ be a positive semidefinite $n \times n$ real matrix admitting a
weighting.  Then  
\[
\mg{\M}
=
\sup_{\x \in \R^n \csuch \x^\transp \M \x \neq 0} 
\frac{\bigl( \sum_{i = 1}^n x_i \bigr)^2}{\x^\transp \M \x}
> 
0.
\]
If $\M$ is positive definite then the supremum is attained by exactly the 
nonzero scalar multiples $\x$ of the unique weighting on $\M$.
\end{lemma}

\begin{proof}
This is a small extension of \cite[Proposition~2.4.3]{MMS}.  Choose a
weighting $\w$ on $\M$.  By the Cauchy--Schwarz inequality,
\[
(\x^\transp \M \w)^2 \leq (\x^\transp \M \x)(\w^\transp \M \w),
\]
or equivalently
\begin{equation}
\label{eq:cs-mult}
\Bigl( \sum x_i \Bigr)^2 \leq (\x^\transp \M \x) \mg{\M},
\end{equation}
for all $\x \in \R^n$.  Equality holds when $\x$ is a scalar multiple of
$\w$, and if $\M$ is positive definite, it holds only then.  Finally,
taking $\x = (1, 0, \ldots, 0)^\transp$ in~\eqref{eq:cs-mult} and using
positive semidefiniteness gives $\mg{\M} > 0$.
\end{proof}

From this, we deduce:

\begin{lemma}
\label{lemma:pos-def-sub}
Let $B \subsetneqq \{1, \ldots, n\}$.  If $\Z$ is positive semidefinite and
both $\Z$ and $\Z_B$ admit a weighting, then $\mg{\Z_B} \leq \mg{\Z}$.
Moreover, if $\Z$ is positive definite and the unique weighting on $\Z$ has
full support, then $\mg{\Z_B} < \mg{\Z}$.
\end{lemma}

\begin{proof}
The first statement follows from Lemma~\ref{lemma:pos-def-variational} and
the fact that $\Z_B$ is positive semidefinite.  The second is trivial if $B
= \emptyset$.  Assuming not, let $\vecstyle{y} \in \R^B$ be the unique
weighting on $\Z_B$ (which is positive definite), and write $\x \in \R^n$
for the extension of $\vecstyle{y}$ by zero to $\{1, \ldots, n\}$.  Then
$\vecstyle{y} \neq \vecstyle{0}$, $\x \neq \vecstyle{0}$, and
\[
\mg{\Z_B}
=
\frac{\bigl(\sum_{i \in B} y_i\bigr)^2}%
{\vecstyle{y}^\transp \Z_B \vecstyle{y}}
=
\frac{\bigl(\sum_{i = 1}^n x_i\bigr)^2}%
{\x^\transp \Z \x}.
\]
But $\x$ does not have full support, so by hypothesis, it is not a scalar
multiple of the unique weighting on $\Z$.  Hence by
Lemma~\ref{lemma:pos-def-variational}, $(\sum x_i)^2/\x^\transp \Z \x <
\mg{\Z}$.
\end{proof}

We now apply this result on magnitude to the maximization of diversity.

\begin{propn}
\label{propn:pos-def-max}
Suppose that $\Z$ is positive semidefinite.  If $\Z$ has a nonnegative
weighting $\w$, then $\Dmax{\Z} = \mg{\Z}$ and $\w/\mg{\Z}$ is a maximizing
distribution.  Moreover, if $\Z$ is positive definite and its unique
weighting $\w$ is positive then $\w/\mg{\Z}$ is the unique maximizing
distribution.
\end{propn}

\begin{proof}
This follows from Theorem~\ref{thm:comp} and
Lemma~\ref{lemma:pos-def-sub}. 
\end{proof}

In particular, if $\Z$ is positive semidefinite and has a nonnegative
weighting, then its maximum diversity can be computed in polynomial time.

\begin{cor}
\label{cor:pos-def-support}
If $\Z$ is positive definite with positive weighting, then its unique
maximizing distribution has full support.  \qed
\end{cor}

In other words, when $\Z$ has these properties, its maximizing distribution
eliminates no species.  Here are three classes of such matrices $\Z$.

\begin{example}
\label{eg:ultra}
Call $\Z$ \demph{ultrametric} if $Z_{ik} \geq \min \{Z_{ij}, Z_{jk}\}$ for
all $i, j, k$ and $Z_{ii} > \max_{j \neq k} Z_{jk}$ for all $i$.  (Under
the assumptions~\eqref{eq:stronger-sim-hyps} on $\Z$, the latter condition
just states that distinct species are not completely similar.)  If $\Z$ is
ultrametric then $\Z$ is positive definite with positive weighting, by
\cite[Proposition~2.4.18]{MMS}.

Such matrices arise in practice: for instance, $\Z$ is ultrametric if it is
defined from a phylogenetic or taxonomic tree as in
Examples~\ref{eg:sim-phylo} and~\ref{eg:sim-taxo}.
\end{example}

\begin{example}
Let $\rvec \in \Delta_n$ be a probability distribution of full support,
and write $\Z$ for the diagonal matrix with entries $1/r_1, \ldots, 1/r_n$.
Then for $0 < q < \infty$,
\[
-\log \dqz(\p) 
=
\begin{cases}
\frac{1}{q - 1} \log \sum_{i \in \supp(\p)} p_i^q r_i^{1 - q}        &
\text{if } q \neq 1,    \\
\sum_{i \in \supp(\p)} p_i \log (p_i/r_i)       &
\text{if } q = 1.
\end{cases}
\]
The right-hand side is the \demph{R\'enyi relative entropy} or
\demph{R\'enyi divergence} $\rre{q}{\p}{\rvec}$ \cite[Section~3]{Reny}.
Evidently $\Z$ is positive definite, and its unique weighting $\rvec$ is
positive.  (In fact, $\Z$ is ultrametric.)  So
Proposition~\ref{propn:pos-def-max} applies; in fact, it gives the
classical result that $\rre{q}{\p}{\rvec} \geq 0$ with equality if and only
if $\p = \rvec$.
\end{example}

\begin{example}
\label{eg:id-pos-def}
The identity matrix $\Z = \I$ is certainly positive definite with positive
weighting.  By topological arguments, there is a neighbourhood $U$ of $\I$
in the space of symmetric matrices such that every matrix in $U$ also has
these properties.  (See the proofs of \cite[Propositions~2.2.6
  and~2.4.6]{MMS}.)  Quantitative versions of this result are also
available.  For instance, in~\cite[Proposition~2.4.17]{MMS} it was shown
that $\Z$ is positive definite with positive weighting if $Z_{ii} = 1$ for
all $i$ and $Z_{ij} < 1/(n - 1)$ for all $i \neq j$.  In fact, this result
can be improved:
\end{example}

\begin{propn}
Suppose that $Z_{ii} = 1$ for all $i, j$ and that $\Z$ is \demph{strictly
  diagonally dominant} (that is, $Z_{ii} > \sum_{j \neq i} Z_{ij}$ for all
$i$).  Then $\Z$ is positive definite with positive weighting.
\end{propn}

\begin{proof}
Since $\Z$ is real symmetric, it is diagonalizable with real eigenvalues.
By the hypotheses on $\Z$ and the Gershgorin disc theorem~\cite[Theorem
  6.1.1]{HoJoMA}, every eigenvalue of $\Z$ is in the interval $(0, 2)$.  It
follows that $\Z$ is positive definite and that every eigenvalue of $\I -
\Z$ is in $(-1, 1)$.  Hence $\I - \Z$ is similar to a diagonal matrix with
entries in $(-1, 1)$, and so $\sum_{k = 0}^\infty (\I - \Z)^k$ converges to
$(\I - (\I - \Z))^{-1} = \Z^{-1}$.  Thus,
\begin{equation}
\label{eq:gersh-inv}
\Z^{-1}
=
\sum_{k = 0}^\infty (\I - \Z)^{k}
=
\sum_{k = 0}^\infty (\Z - \I)^{2k} (2\I - \Z).
\end{equation}
Writing $\e = (1\ \cdots\ 1)^\transp$, the unique weighting on $\Z$ is $\w
= \Z^{-1}\e$.  The hypotheses on $\Z$ imply that $\Z - \I$ has nonnegative
entries and $(2\I - \Z)\e$ has positive entries.  Hence
by~\eqref{eq:gersh-inv},
\[
\w = \Z^{-1}\e \geq (\Z - \I)^0 (2\I - \Z) \e = (2\I - \Z)\e
\]
entrywise, and so $\w$ is positive.
\end{proof}

Thus, a matrix $\Z$ that is ultrametric, or
satisfies~\eqref{eq:stronger-sim-hyps} and is strictly diagonally dominant,
has many special properties: the maximum diversity is equal to the
magnitude, there is a unique maximizing distribution, the maximizing
distribution has full support, and both the maximizing distribution and the
maximum diversity can be computed in polynomial time.

\section{Preservation of species}
\label{sec:elim}

We saw in Examples~\ref{eg:graph-lin3} and~\ref{eg:graph-lin4} that for
certain similarity matrices $\Z$, none of the maximizing distributions has
full support.  Mathematically, this simply means that maximizing
distributions sometimes lie on the boundary of $\Delta_n$.  But
ecologically, it may sound shocking: is it reasonable that diversity can be
increased by \emph{eliminating} some species?

We argue that it is.  Consider, for instance, a forest consisting of one
species of oak and ten species of pine, with each species equally abundant.
Suppose that an eleventh species of pine is added, again with equal
abundance (Figure~\ref{fig:pine}).  This makes the forest even more heavily
dominated by pine, so it is intuitively reasonable that the diversity
should decrease.  But now running time backwards, the conclusion is that if
we start with a forest containing the oak and all eleven pine species,
eliminating the eleventh should \emph{increase} diversity.

\begin{figure}
\begin{center}
\setlength{\unitlength}{0.7mm}
\begin{picture}(110,31)
\cell{55}{-1}{b}{\includegraphics[width=110\unitlength]{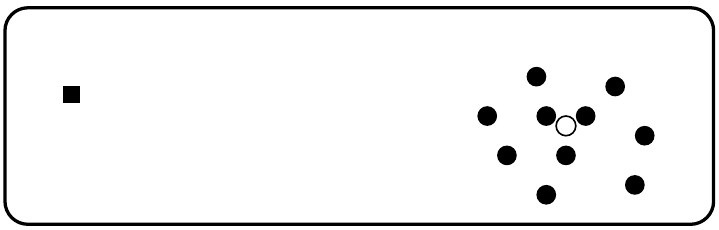}}
\cell{11.2}{24.5}{c}{oak}
\cell{87}{28.5}{c}{pine}
\end{picture}%
\end{center}
\caption{Hypothetical community consisting of one species of oak
  ($\scriptscriptstyle\blacksquare$) and ten species of pine ($\bullet$), to
  which one further species of pine is then added ($\circ$).  Distances
  between species indicate degrees of dissimilarity (not to scale).}
\label{fig:pine}
\end{figure}

To clarify further, recall from Section~\ref{sec:model} that diversity is
defined in terms of the \emph{relative} abundances only.  Thus, eliminating
species $i$ causes not only a decrease in $p_i$, but also an increase in
the other relative abundances $p_j$.  If the $i$th species is particularly
ordinary within the community (like the eleventh species of pine), then
eliminating it makes way for less ordinary species, resulting in a more
diverse community.

The instinct that maximizing diversity should not eliminate any species is
based on the assumption that the distinction between species is of high
value.  (After all, if two species were very nearly identical---or in the
extreme, actually identical---then losing one would be of little
importance.)  If one wishes to make that assumption, one must build it into
the model.  This is done by choosing a similarity matrix $\Z$ with a low
similarity coefficient $Z_{ij}$ for each $i \neq j$.  Thus, $\Z$ is close
to the identity matrix $\I$ (assuming that similarity is measured on a
scale of $0$ to $1$).  Example~\ref{eg:id-pos-def} guarantees that in this
case, there is a unique maximizing distribution and it does not, in fact,
eliminate any species.

(The fact that maximizing distributions can eliminate some species has
previously been discussed in the ecological literature in the case $q = 2$;
see Pavoine and Bonsall~\cite{PaBo} and references therein.)

We now derive necessary and sufficient conditions for a similarity matrix
$\Z$ to admit at least one maximizing distribution of full support, and
also necessary and sufficient conditions for \emph{every} maximizing
distribution to have full support.  The latter conditions are genuinely more
restrictive; for instance, if $\Z = \bigl( \begin{smallmatrix} 1& 1\\ 1&
  1\\ \end{smallmatrix} \bigr)$ then some but not all maximizing
distributions have full support.

\begin{lemma}
\label{lemma:elim-main}
If at least one maximizing distribution for $\Z$ has full support then $\Z$
is positive semidefinite and admits a positive weighting.  Moreover, if
every maximizing distribution for $\Z$ has full support then $\Z$ is
positive definite.
\end{lemma}

\begin{proof}
Fix a maximizing distribution $\p$ of full support.  Maximizing
distributions are invariant (Corollary~\ref{cor:max-invt}), so by
\bref{part:invt-invt}$\implies$\bref{part:invt-ext} of
Lemma~\ref{lemma:invt}, $\mg{\Z}\p$ is a weighting of $\Z$ and $\mg{\Z} >
0$.  In particular, $\Z$ has a positive weighting.

Now we imitate the proof of Proposition~3B of~\cite{FrTa}.  For each $\s \in
\R^n$ such that $\sum_{i = 1}^n s_i = 0$, define a function $f_\s: \R \to
\R$ by
\[
f_\s(t) = (\p + t\s)^\transp \Z (\p + t\s).
\]
Using the symmetry of $\Z$ and the fact that $\mg{\Z}\p$ is a weighting, we
obtain
\begin{align}
f_\s(t)  &
= 
\p^\transp \Z \p + 2\s^\transp \Z \p \cdot t + \s^\transp \Z \s \cdot t^2
\nonumber
\\
&
=
1/\mg{\Z} + \s^\transp \Z \s \cdot t^2.
\label{eq:fs-final}
\end{align}
Now $\sum s_i = 0$ and $\p$ has full support, so $\p + t\s \in \Delta_n$
for all real $t$ sufficiently close to zero.  But $f_\s(t) =
1/\dvv{2}{\Z}(\p + t\s)$ for such $t$, so $f_\s$ has a local minimum at
$0$.  Hence $\s^\transp \Z \s \geq 0$.  It follows that $f_\s$ is
everywhere positive.

We have shown that $\s^\transp \Z \s \geq 0$ whenever $\s \in \R^n$ with
$\sum s_i = 0$.  Now take $\x \in \R^n$ with $\sum x_i \neq 0$.
Put $\s = \x/\bigl(\sum x_i\bigr) - \p$.  Then $\sum s_i = 0$, and
\begin{equation}
\label{eq:not-perp-pos}
\x^\transp \Z \x 
=
\Bigl( \sum x_i \Bigr)^2 f_\s(1)
>
0.
\end{equation}
Hence $\Z$ is positive semidefinite.  

For `moreover', assume that every maximizing distribution for $\Z$ has full
support.  By~\eqref{eq:not-perp-pos}, we need only show that $\s^\transp \Z
\s > 0$ whenever $\s \neq \vecstyle{0}$ with $\sum s_i = 0$.  Given such an
$\s$, choose $t \in \R$ such that $\p + t\s$ lies on the boundary of
$\Delta_n$.  Then $\p + t\s$ does not have full support, so is not
maximizing, so does not maximize $\dvv{2}{\Z}$ (by
Corollary~\ref{cor:irrelevance}).  Hence $f_\s(t) > f_\s(0)$, which
by~\eqref{eq:fs-final} implies that $\s^\transp \Z \s > 0$.
\end{proof}

We can now prove the two main results of this section.

\begin{propn}
\label{propn:semi}
The following are equivalent:
\begin{enumerate}
\item
\label{part:semivec}
there exists a maximizing distribution for $\Z$ of full support;

\item
\label{part:semidef}
$\Z$ is positive semidefinite and admits a positive weighting.
\end{enumerate}
\end{propn}

\begin{proof}
\bref{part:semivec}$\implies$\bref{part:semidef} is the first part of
Lemma~\ref{lemma:elim-main}.  For the converse, assume~\bref{part:semidef}
and choose a positive weighting $\w$.  Then $\mg{\Z} > 0$, so $\p =
\w/\mg{\Z}$ is a probability distribution of full support.  We have
$\dqz(\p) = \mg{\Z}$ for all $q$, by Lemma~\ref{lemma:invt}.  But the
computation theorem implies that $\Dmax{\Z} = \mg{\Z_B}$ for some $B \sub
\{1, \ldots, n\}$ such that $\Z_B$ admits a weighting, so $\Dmax{\Z} \leq
\mg{\Z}$ by Lemma~\ref{lemma:pos-def-sub}.  Hence $\p$ is maximizing.
\end{proof}

\begin{propn}
The following are equivalent:
\begin{enumerate}
\item
\label{part:vec}
every maximizing distribution for $\Z$ has full support;

\item
\label{part:unique}
$\Z$ has exactly one maximizing distribution, which has full support;

\item
\label{part:def}
$\Z$ is positive definite with positive weighting;

\item 
\label{part:crit}
$\Dmax{\Z} > \Dmax{\Z_B}$ for every nonempty proper subset $B$ of $\{1,
  \ldots, n\}$.
\end{enumerate}
\end{propn}

(The weak inequality $\Dmax{\Z} \geq \Dmax{\Z_B}$ holds for any $\Z$, by
the absent species lemma (Lemma~\ref{lemma:abs}).)

\begin{proof}
\bref{part:vec}$\implies$\bref{part:def} and
\bref{part:def}$\implies$\bref{part:unique} are immediate from
Lemma~\ref{lemma:elim-main} and Proposition~\ref{propn:pos-def-max}
respectively, while \bref{part:unique}$\implies$\bref{part:vec} is trivial.

For \bref{part:vec}$\implies$\bref{part:crit}, assume~\bref{part:vec}.  Let
$\emptyset \neq B \subsetneqq \{1, \ldots, n\}$.  Choose a maximizing
distribution $\p'$ for $\Z_B$, and denote by $\p$ its extension by zero to
$\{1, \ldots, n\}$.  Then $\p$ does not have full support, so
there is some $q \in [0, \infty]$ such that $\p$ fails to maximize $\dqz$.
Hence
\[
\Dmax{\Z_B} = \dvv{q}{\Z_B}(\p') = \dvv{q}{\Z}(\p) < \Dmax{\Z},
\]
where the second equality is by the absent species lemma.

For \bref{part:crit}$\implies$\bref{part:vec}, assume~\bref{part:crit}.
Let $\p$ be a maximizing distribution for $\Z$.  Write $B = \supp(\p)$, and
denote by $\p'$ the restriction of $\p$ to $B$.  Then for any $q$,
\[
\Dmax{\Z_B} \geq \dvv{q}{\Z_B}(\p') = \dvv{q}{\Z}(\p) = \Dmax{\Z},
\]
again by the absent species lemma.  Hence by~\bref{part:crit}, $B = \{1,
\ldots, n\}$.
\end{proof}

\section{Open questions}
\label{sec:questions}

The main theorem, the computation theorem and
Corollary~\ref{cor:irrelevance} answer all the principal questions about
maximizing the diversities~$\dqz$.  Nevertheless, certain questions remain.

First, there are computational questions.  We have found two classes of
matrix $\Z$ for which the maximum diversity and maximizing distributions
can be computed in polynomial time: ultrametric matrices
(Example~\ref{eg:ultra}) and those close to the identity matrix $\I$
(Example~\ref{eg:id-pos-def}).  Both are biologically significant.  Are
there other classes of similarity matrix for which the computation can be
performed in less than exponential time?

Second, we may seek results on maximization of $\dqz(\p)$ under constraints
on $\p$.  There are certainly some types of constraint under which both
parts of Theorem~\ref{thm:main} fail, for trivial reasons: if we choose two
distributions $\p$ and $\p'$ whose diversity profiles cross
(Figure~\ref{fig:main}(b)) and constrain our distribution to lie in the set
$\{\p, \p'\}$, then there is no distribution that maximizes $\dqz$ for all
$q$ simultaneously, and the maximum value of $\dqz$ also depends on $q$.
But are there other types of constraint under which the main theorem still
holds?

In particular, the distribution might be constrained to lie close to a
given distribution $\p$.  The question then becomes: if we start with a
distribution $\p$ and have the resources to change it by only a given small
amount, what should we do in order to maximize the diversity?

Third, we have seen that every symmetric matrix $\Z$
satisfying~\eqref{eq:weaker-sim-hyps} (for instance, every symmetric matrix
of positive reals) has attached to it a real number, the maximum diversity
$\Dmax{\Z}$.  What is the significance of this invariant?

We know that it is closely related to the magnitude of matrices.  This has
been most intensively studied in the context of metric spaces.  By
definition, the magnitude of a finite metric space $X$ is the magnitude of
the matrix $\Z = (e^{-d(i, j)})_{i, j \in X}$; see
\cite{MMS,AMSES,MeckPDM}, for instance.  In the metric context, the meaning
of magnitude becomes clearer after one extends the definition from finite
to compact spaces (which is done by approximating them by finite
subspaces).  Magnitude for compact metric spaces has recognizable geometric
content: for example, the magnitude of a 3-dimensional ball is a cubic
polynomial in its radius \cite[Theorem~2]{BaCa}, and the magnitude of a
homogeneous Riemannian manifold is closely related to its total scalar
curvature~\cite[Theorem~11]{WillMSS}.

Thus, it is natural to ask: can one extend Theorem~\ref{thm:main} to some
class of `infinite matrices' $\Z$?  (For instance, $\Z$ might be the form
$(x, y) \mapsto e^{-d(x, y)}$ arising from a compact metric space.  In this
case, the maximum diversity of order $2$ is a kind of capacity, analogous
to classical definitions in potential theory; for a compact subset of
$\R^n$, it coincides with the Bessel capacity of an appropriate
order~\cite{MeckMDC}.)  And if so, what is the geometric significance of
maximum diversity in that context?

There is already evidence that this is a fruitful line of enquiry.
In~\cite{MeckMDC}, 
Meckes gave a definition of the maximum diversity of order $2$ of a compact
metric space, and used it to prove a purely geometric theorem relating
magnitude to fractional dimensions of subsets of $\R^n$.  If this maximum
diversity can be shown to be equal to the maximum diversity of all other
orders then further geometric results may come within reach.

The fourth and final question concerns interpretation.  Throughout, we have
interpreted $\dqz(\p)$ in terms of ecological diversity.  However, there is
nothing intrinsically biological about any of our results.  For example, in
an information-theoretic context, the `species' might be the code symbols,
with two symbols seen as similar if one is easily mistaken for the other;
or if one wishes to transmit an image, the `species' might be the colours,
with two colours seen as similar if one is an acceptable substitute for the
other (much as in rate distortion theory~\cite{CoTh}).  Under these or
other interpretations, what is the significance of the theorem that
the diversities of all orders can be maximized simultaneously?

\paragraph{Acknowledgements}
We thank Christina Cobbold, Ciaran McCreesh and Richard Reeve for helpful
discussions.  This work was supported by the Carnegie Trust for the
Universities of Scotland, the Centre de Recerca Matem\`atica, an EPSRC
Advanced Research Fellowship, the National Institute for Mathematical and
Biological Synthesis, the National Science Foundation, and the Simons
Foundation.

\end{document}